\documentclass[letterpaper]{article} 
\usepackage{aaai25}  
\usepackage{times}  
\usepackage{helvet}  
\usepackage{courier}  
\usepackage[hyphens]{url}  
\usepackage{graphicx} 
\usepackage{natbib}  
\usepackage{caption} 
\usepackage{algorithm}
\usepackage{algorithmic}

\usepackage{newfloat}
\usepackage{listings}
\DeclareCaptionStyle{ruled}{labelfont=normalfont,labelsep=colon,strut=off} 
\floatstyle{ruled}
\newfloat{listing}{tb}{lst}{}
\floatname{listing}{Listing}
\title{On the Power of Strategic Corpus Enrichment in Content Creation Games}

\author{
    Haya Nachimovsky, Moshe Tennenholtz \\
}
\affiliations {
    Technion - Israel Institute of Technology, Haifa, Israel \\
    haya.nac@campus.technion.ac.il, moshet@technion.ac.il
}

\usepackage{graphicx}
\usepackage{tikz}
\usetikzlibrary{positioning} 
\usepackage{amsmath}
\usepackage{amssymb}
\usepackage{amsthm}
\usepackage{dsfont}
\usepackage{xcolor}
\usepackage{mathtools}
\usepackage{xspace}
\usepackage{thmtools}
\usepackage{subcaption}

\theoremstyle{plain}
\newtheorem{theorem}{Theorem}[section]
\newtheorem{corollary}[theorem]{Corollary}

\theoremstyle{definition}
\newtheorem{definition}{Definition}

\newcommand{\floor}[1]{\left\lfloor #1 \right\rfloor}
\newcommand{\ceil}[1]{\left\lceil #1 \right\rceil}

\newcommand{\definedas}{\coloneq}

\newcommand{\omt}[1]{}
\newcommand{\firstmention}[1]{\emph{#1}}

\newcommand{\pFloor}{\floor{\frac{1}{p}}}
\newcommand{\plantedDocsSet}{T}

\newcommand{\generalThreshold}{\vec{t}}

\newcommand{\fairPlayer}{fair}
\newcommand{\devTypeZ}{type 1\xspace}
\newcommand{\devTypeZplusOne}{type 2\xspace}
\newcommand{\solo}{solo}

\newcommand{\indicator}[1]{\mathds{1}_{\{#1\}}}

\newcommand{\mymod}[1]{\ (\mathrm{mod}\ #1)}

\begin{document}

\maketitle

\begin{abstract}
Search and recommendation ecosystems exhibit competition among content creators. This competition has been tackled in a variety of game-theoretic frameworks. Content creators generate documents with the aim of being recommended by a content ranker for various information needs. In order for the ecosystem, modeled as a content ranking game, to be effective and maximize user welfare, it should guarantee stability, where stability is associated with the existence of pure Nash equilibrium in the corresponding game. Moreover, if the contents' ranking algorithm possesses a game in which any best-response learning dynamics of the content creators converge to equilibrium of high welfare, the system is considered highly attractive. However, as classical content ranking algorithms, employed by search and recommendation systems, rank documents by their distance to information needs, it has been shown that they fail to provide such stability properties. As a result, novel content ranking algorithms have been devised. In this work, we offer an alternative approach: corpus enrichment with a small set of fixed dummy documents. It turns out that, with the right design, such enrichment can lead to pure Nash equilibrium and even to the convergence of any best-response dynamics to a high welfare result, where we still employ the classical/current content ranking approach. We show two such corpus enrichment techniques with tight bounds on the number of documents needed to obtain the desired results. Interestingly, our study is a novel extension of Borel's Colonel Blotto game.
\end{abstract}

\section{Introduction}

Content creators aim that their content will be recommended by search and recommendation systems when faced with information needs expressed either implicitly or explicitly by the users \cite{compSearch}. In abstract terms, a content creator (publisher) creates a content item (document) with the aim of being selected by a content ranking algorithm to be shown to a user with a particular topic of interest. At first approximation, a content creator aims that his document will be selected for as many topics (e.g., relevant queries) taken from a given set \cite{Haya1arxiv}.

This general setup has been widely acknowledged by the AI community and various models have been advocated. 
In the line of work we consider, content creators aim to promote their content rather than serve as malicious attackers in the system (as in, e.g., \cite{Chen+al:23a,Chen+al:23b,Liu+al:23a}).  
Existing competitive content creation models emphasize the search ecosystem \cite{raifer2017information} and the recommendation ecosystem \cite{ben2018game, FacilityLocationRec}.
Different models also vary in their representation of users and content, which can be explicit~\cite{YaoLSLZWWX23,EilatR23,Jagadeesan0S23,HronKJKD23} or via queries~\cite{Goren_2021}. 
Others discuss a more general notion of topics of interest~\cite{basat2017game}.
However, all call for game-theoretic analysis of strategic content creation games, where publishers aim that some content ranking algorithm will select their content.  

One property that is crucial for the above efforts is the general observation that in order to lead to high user social welfare, we need to care for publishers' utility \cite{mladenov2020optimizing}, and that in order to maintain publishers, the game should converge to some stable deterministic equilibrium~\cite{basat2017game,YaoLSLZWWX23,ben2018game, LearnignDynamicsPaper}. The related stability concept is captured by the requirement that the corresponding content ranking algorithm determines a game with pure Nash equilibrium, a deterministic content profile where unilateral deviations are not beneficial. Moreover, if that equilibrium can be shown to be obtained under best response learning dynamics, the system is especially appealing as any reasonable behavior (starting with any contents and with arbitrary order of modifications) will result in such a stable situation.

Given the above, a central theme of foundational work on strategic recommendation and search ecosystems originates from the fact that existing content ranking algorithms are unstable when facing such competition. Namely, given an information need, recommending the content closest to it (according to some well-crafted distance function estimating relevance) determines a game that does not possess a pure Nash equilibrium, hence might lead to harmful fluctuations.  Indeed, most recent literature on strategic search and recommendation ecosystems attempts to offer alternative content ranking algorithms that do provide such stability \cite{YaoLSLZWWX23,ben2018game,basat2017game}.

One may wonder whether stability can be obtained if we do not modify the content ranking algorithms and stick to the well-crafted classical ones. 
In this work, we introduce and illustrate such general alternative: {\em strategic corpus enrichment}. 
To initiate the study of this approach, we use a model introduced in~\citet{Haya1arxiv}, expanding on studies in~\citet{raifer2017information}) and in~\citet{basat2017game}. 

Given a set of possible topics of interest (queries), each publisher needs to create a fixed-length document with content split among the topics. The content ranker provides a score for each publisher for each topic, where this score is a non-decreasing function of the weight the provider assigns to that topic (while this is typical for most existing systems, one can think about the proportion of news dedicated to the topic or the number of words associated with the corresponding query as naive estimates). For each topic, the publisher with the highest score is selected, with a uniform tie-breaking rule. A publisher strategy is the split of weights among topics, and its utility is the sum of topics it wins \cite{Haya1arxiv}.
While in recommendation systems the model is illustrative for current strategic content creation games, we observe that this model is an instance of one of the most celebrated economic models: the Colonel Blotto game, going back to~\citet{borel1921}.
In this setting, there is a set of battlefields and each colonel needs to split a given fixed budget (army troops) between battlefields. A colonel wins a battlefield if his army is the largest, while aiming to win as many battlefields as possible. Hence, our content creation games are Colonel Blotto games among publishers, with queries as battlefields.
Failures of the classical content ranking algorithm to obtain stable behavior are a direct implication of the lack of pure Nash equilibrium in Colonel Blotto games.

Nevertheless, in current AI ecosystems, we have a possibility that has not been considered so far: the (artificial) generation of static documents. 
If we add a relatively small number of static documents we modify the game. 
If this can be done in a way leading to pure Nash equilibrium where only the original documents win and (original) publishers even reach that equilibrium under any best-response learning dynamics, while obtaining high welfare, we can reach a stable and effective ecosystem while using the classical well-studied content ranking algorithm!

The above is exactly what we manage to obtain in our work.
We begin by introducing the \emph{uniform corpus extension technique} and show that it leads to our desired results. 
This approach reveals that irrespective of the number of topics or the classical content ranker employed, we can guarantee a pure Nash equilibrium and the convergence of any best-response learning dynamics to a welfare-maximizing result by at most doubling the number of documents. Furthermore, our findings establish tight bounds on the necessary number of documents under the problem parameters.  
Building on this foundation, we introduce the \emph{generalized corpus extension technique}, which further reduces the number of documents required to achieve pure Nash equilibrium. However, in that context, best response dynamics do not always converge (so the system itself should recommend to publishers how to split documents among topics to lead to stability).
Our results underscore the effectiveness of using strategic corpus enrichment to achieve stability.

To summarize, our contributions are as follows:
\begin{itemize}
    \item We propose the method of corpus enrichment as a novel strategy for mediators to influence publisher behavior within search and recommendation systems.
    \item We demonstrate that our approach consistently achieves stability (equilibrium), addressing the limitations of prior approaches where pure equilibria rarely exist. Furthermore, we provide tight bounds on the minimum number of documents needed in the corpus enrichment technique to ensure stability.
    \item Our equilibrium profiles not only ensure that publisher exposure remains unaffected but also maximize user social welfare, i.e., they maximize both publisher exposure and user social welfare. 
    \item
    Our method guarantees convergence of best-response dynamics, while without corpus enrichment, they fail to converge despite the presence of a pure equilibrium.
    Additionally, we propose a computationally efficient algorithm for calculating the best response strategies.
\end{itemize}

\section{Model}
\label{sec:model}

\subsection{Multiple-Queries Ranking Game}
Our model is based on the framework introduced by~\citet{Haya1arxiv} for ranking competitions between publishers aiming to be ranked highly for multiple queries. 
The \emph{Multiple-Queries Ranking Game} $G =\langle n, m, p \rangle$ is defined as a $n$-player's game, 
where each player is a publisher who writes and modifies a single document.
Let $Q$ be a query set that contains $m$ queries: $Q = \{q_1,\ldots,q_m\}$.
The term queries may refer to arbitrary content types or to different query aspects. 
A document $d$ is defined as an element in $D \definedas \{(d^1,\ldots,d^m) \in[0,1]^m : \sum_{j=1}^m d^j \le 1\}$.
Each component $d^j$ ($\in [0,1]$) signifies the degree to which the document $d$ focuses on the query $q_j$.

Let $f: [0,1] \rightarrow [0,1]$ be a single peak function, meaning there exists a unique $p \in [0,1]$ such that $f$ is non-decreasing in $[0,p]$ and non-increasing in $[p,1]$.
The retrieval function $r: D \times Q \rightarrow [0,1]$ assigns each document $d~(\in D$) a score with respect to a query $q_j$: $r(d, q_j) \definedas f(d^j)$.
The documents are then ranked in descending order of scores (assuming ties are broken arbitrarily).
Notice that if the function is constant in $[p,1]$ (or in particular, $p=1$) then we have here the classical selection of the document as most relevant to the query. We allow the function to be non-increasing in $[p,1]$ to deal with issues such as keyword stuffing (preventing spam by overloading query terms). None of our results will change given that assumption.

The document $d_i$ published by player $i$ is considered the pure \firstmention{strategy} of player $i$. 
This strategy involves distributing the document's focus across the various queries.
A \firstmention{strategy profile} $s \definedas (d_1,\ldots,d_n)$ is the tuple of pure strategies of the players in the game.
We use $s_{-i} \definedas (d_1, \ldots, d_{i-1}, d_{i+1}, \ldots, d_n)$ to denote the strategies of all players except player $i$.

The \firstmention{utility} of a player $i$ who writes document $d_i$, denoted $U_i(s)$, is defined as the number of queries where $d_i$ achieves the highest rank among all documents. 
A strategy profile $s$ is a (pure) \firstmention{Nash equilibrium} if no player $i$ has an incentive to change their strategy given the strategies $s_{-i}$ of all other players. In other words, modifying their strategy does not lead to higher utility.

\subsection{Corpus-Enriched Ranking Game}

Adopting game-theoretic terminology, we will refer to the content selection mechanism (e.g. search engine, recommendation, system) as a mediator. 
The mediator's goal is to promote the desired behavior of publishers, such as stability and satisfaction of the users.
We propose a novel approach for the mediator to achieve this goal: adding a set of static documents $T$ to the corpus.
With these static documents, the rank of each of the publisher's documents $d_i$ for a query depends not only on the strategies of all other players but also on the static documents. 

The \emph{Corpus-Enriched Ranking Game} is defined as $G_{\plantedDocsSet} = \langle n, m, p, \plantedDocsSet \rangle$. In the game $G_{\plantedDocsSet}$, the rank of each of the publisher's documents $d_i$ for a query depends not only on the strategies of all other players but also on the static documents $\plantedDocsSet$.

Note that publishers derive utility solely from the highest-ranked document for a given query, as any lower ranking yields zero utility.
Furthermore, publishers are indifferent to the identity of the highest-ranked document or any other documents that rank lower within the set $T$.
We can therefore define a simplified version of the game $G_{\plantedDocsSet}$.

\subsection{Simplified Corpus-Enriched Ranking Game}

Let $\generalThreshold_j$ denote the value of the document in $T$ corresponding to the highest ranking score with respect to query $q_j\in Q$. 
The impact of $T$ on publishers' behavior can be captured solely by $\generalThreshold$, where $\generalThreshold_j$ essentially denotes a \emph{threshold} that any publisher must surpass to achieve the highest rank for that query.

We define the \emph{Ranking Game with Thresholds} $G_{\generalThreshold} =\langle n, m, p, \generalThreshold \rangle$ with $\generalThreshold \in [0,1]^m$, where the static documents are represented using a single vector of thresholds. 
Thresholds in our model function similarly to reserved prices in auctions.
In auctions, enforcing the reserved price does not necessitate physical bids but can be achieved by automatically rejecting any bids that fall below this threshold. 
Similarly, the mediator could set a publicly known threshold for each query, ensuring that a document cannot be considered the winner if its score is below this threshold.

\begin{restatable}{lemma}{Translation}
\label{lem:translation}
    For a given $\eta \in \mathbb{N}$, there exists a document set $T$ such that $|T| = \eta$ and the game $G_T = \langle n, m, p, T \rangle$ has a pure equilibrium only if there exists a vector $\Vec{t} \in [0,1]^m$ such that $\|\Vec{t} \|_1 \le \eta$ and the game $G_{\Vec{t}} = \langle n, m, p, \Vec{t} \rangle$ has a pure equilibrium. 
\end{restatable}

All proofs of lemmas, theorems, and corollaries presented in this paper are detailed in the supplementary material.
The study of the simplified corpus enrichment game will determine the structure of the documents needed for effective corpus enrichment. This will be used to optimize the number of documents needed to obtain equilibrium. Each such document corresponds to a subset of queries where the sum of thresholds of each query in the subset is no more than $1$. 

We specifically aim at corpus enrichment leading to an equilibrium in which only documents published by original publishers emerge as winners, that is, in all queries, there is at least one publisher who scores higher than all static documents.
In the case of a tie between a publisher's document and a static document, we assume that the mediator favors the document published by a real player, and ties between publishers' documents are broken uniformly.
Therefore, we require $\max_{1\le i\le n} r(d_i, q_j) \ge \generalThreshold_j$ for every query $q_j \in Q$. 

\section{Preliminaries}
\label{sec:preliminaries}

The following theorem provides a full characterization of pure equilibrium existence in the game without corpus extension. 

\begin{theorem}[\citet{Haya1arxiv}]
\label{thm:original-full}
    $G=\langle n, m, p \rangle$ has a pure Nash equilibrium iff $p \le \frac{1}{\max{\{\ceil{\frac{2 \cdot m}{n} - 1}, 1\}}}$.
\end{theorem}

In the case where $n<m$, a pure equilibrium rarely exists. Specifically, a pure equilibrium exists only when the peak value is relatively small. 
Henceforth, we assume $n<m$ since in the case where $m\le n$, a pure equilibrium always exists.

The main idea of the proof of Theorem~\ref{thm:original-full} for the case where $n<m$ was to show that two properties are necessary for a pure equilibrium to exist: 1. the winner’s value must match the peak value across all queries 2. no player is the unique winner in more than one query.

\paragraph{Notations}
Given some strategy profile $s$, denote by $w_j(s)$ the value of the document ranked highest in query $q_j \in Q$.
The number of documents written by original players assigned the same highest retrieval score for query $q_j$ when $s$
is played is denoted by $h_j(s)$.
A player $i$ is called a winner in query $q_j$ if $t \le d_i^j = w_j(s)$.
For every player $i$, denote by $J_i(s)$ the set of queries where $i$ is a winner, i.e.; $w_i(s)=d_i^j$. 
Let $J_i^{\solo}(s)$ and $J_i^{tie}(s)$ denote the sets of queries where $i$ wins uniquely and in a tie, respectively.

\section{Uniform Corpus Enrichment}
\label{sec:uniform}

We begin our analysis by examining a subset of games with uniform $\generalThreshold_j$ across all queries, i.e., there exists some $t \in [0,1]$ such that $\generalThreshold_j=t$ for every $q_j \in Q$. With slight abuse of notation, we denote the game by $G_{t}=\langle n, m, p, t\rangle$; and show that even within this constrained framework, it is possible to attain equilibrium by adding a relatively small number of static documents.

\subsection{Equilibrium}

Based on the results of Theorem~\ref{thm:original-full}, we identify the minimum threshold required to ensure the existence of an equilibrium where the winning value matches the peak for all queries.
In this context, the threshold's purpose is to prevent a player from being able to win in at least one more query. 

\begin{restatable}{theorem}{MediumPeakUniform}
\label{thm:n-players-medium-peak-uniform}
$G_{t}=\langle n, m, p, t\rangle$ with $\frac{1}{\ceil{\frac{2m}{n}} - 1} < p \le \frac{1}{\ceil{\frac{m}{n}}}$ has a pure equilibrium iff
    \begin{equation*}         
        t > p - \frac{p \cdot (\pFloor+1) - 1}{ \ceil{\frac{2m}{n}} - \pFloor}
    \end{equation*}
    Moreover, 
    there exists a pure equilibrium $s^*$ where 
    $w_j(s^*) = p$ for all $q_j \in Q$.
\end{restatable}

The proof idea is to show that the first property must still hold (regardless of $t$), and then derive a bound on the threshold that will prevent a player from having a profitable deviation even when he is the unique winner in more than one query, i.e., the second property doesn’t hold.

For larger peak values, achieving an arrangement in which the winners in all queries consistently reach the peak becomes impossible. 
As a result, any potential equilibrium will inevitably have at least one query in which the winning value is lower than the peak. 
This introduces additional complexities within the game. 

Consider a scenario where a player is one of the winners in a particular query but not the sole winner. 
If the winning value for this query is below the peak value, even the smallest amount of excess resources enables the player to deviate profitably by becoming the sole winner of that query.
In contrast, if the winning value of the query already equals the peak value, then no strategy would allow this player to become the sole winner in that query.
Additionally, if a player is tied in three or more queries where the winning value is below the peak, they invariably have a profitable deviation available, where they can become the sole winner in all but one of these queries.

Despite these additional difficulties, our analysis reveals that there exists a specific threshold level at which achieving equilibrium remains feasible. 
Interestingly, this threshold remains constant and does not vary with $p$ (except for the conditions that define its range). 

\begin{restatable}{theorem}{LargePeakUniform}
\label{thm:n-players-large-peak-uniform}
$G_{t}=\langle n, m, p, t\rangle$ where $\frac{1}{\ceil{\frac{m}{n}}} < p$ has a pure equilibrium iff there exists some $\epsilon>0$ such that
    \begin{equation*}  
        t\ge \begin{cases}
            \left (1-\frac{n}{m} \right) \cdot \frac{n}{m} + \epsilon, & m \pmod{n} = 0 \\
            \frac{1}{\ceil{\frac{m}{n}}}, & \text{otherwise} \\
        \end{cases}
    \end{equation*}
    Additionally, there is a pure equilibrium $s^*$ where $w_j(s^*) = \frac{1}{\ceil{\frac{m}{n}}}$ for all $q_j \in Q$.
\end{restatable}

In the case where the number of queries is exactly divisible by the number of players, there exists some arrangement where there are no ties between the players. 
Specifically, there are no ties where the winning value is lower than the peak. 
However, if $m \pmod{n} > 0$, it is inevitable to have at least one query with more than one winner. 
This situation necessitates additional constraints, as previously mentioned.

We found the minimal threshold required for the existence of an equilibrium for any number of players, queries, and peak value.
It can be concluded that for any game $G_{t}$, the minimum number of documents required to reach equilibrium does not exceed the number of players. 
This finding holds true regardless of the number of queries, even as $m \rightarrow \infty$, indicating that merely doubling the number of documents is sufficient to stabilize the system.

\begin{restatable}{corollary}{NDocsEquilibriumUniform}
Given any number of players $n\in \mathbb{N}$, number of queries $m \in \mathbb{N}$ and peak value $p \in (0,1]$, there exists a document set $\plantedDocsSet$ such that $|\plantedDocsSet|\le n$, for which the game $G=\langle n, m, p, \plantedDocsSet\rangle$ has a pure equilibrium.
\end{restatable}

\subsection{Social Welfare}

\begin{figure}
    \centering
    \begin{subfigure}{0.25\linewidth}
        \centering
        \includegraphics[width=0.9\linewidth]{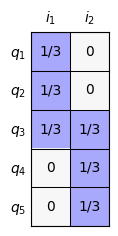}
        \caption{$t=\frac{1}{3}$}
        \label{subfig:social-welfare-example-game1}
    \end{subfigure}%
    \begin{subfigure}{0.25\linewidth}
        \centering
         \includegraphics[width=0.9\linewidth]{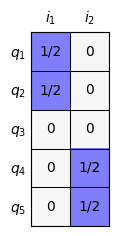}
        \caption{$t=\frac{1}{3} + \epsilon$}
        \label{subfig:social-welfare-example-game2}
    \end{subfigure}
    \caption{Illustration of different equilibrium profiles in the game $G_t=\langle 2, 5, 0.5, t \rangle$ for varying $t$ values.}
    \label{fig:social-welfare-example}

\end{figure}

The publishers' exposure in the equilibria we found is always maximized since only original publishers emerge as winners. 
Interestingly, while one might expect a tradeoff between the publishers' and users' welfare, we show that our proposed method of corpus enrichment can maximize both.

The selection of $\plantedDocsSet$ influences not only the existence of an equilibrium but also its pattern. 
Figure~\ref{fig:social-welfare-example} illustrates an example of a game where even a minimal increase in the threshold value substantially changes the equilibrium profile. 
In the first game (Figure~\ref{subfig:social-welfare-example-game1}) the winning value for all queries is exactly $\frac{1}{3}$. 
In the second game (Figure~\ref{subfig:social-welfare-example-game2}) the winners' scores increased in $4$ out of $5$ queries, which benefits user welfare, but no original publisher wins in $q_3$.

Typical semantics in the literature when discussing multiple queries is of several aspects of the same topic, and the aim is to provide good coverage for all different aspects. 
For example, in the context of the "Paris Olympic Games", we might analyze multiple queries addressing different aspects such as events, venues, and participants.
Social welfare is defined as the distance from the peak value of the worst-performing query to emphasize coverage across all query aspects.
Formally, 
\begin{definition}[Social Welfare]
    \[
    SW(s) = - \max_{q_j \in Q} |p - w_j(s)|
    \]
\end{definition}

We show that the equilibria described in Theorems~\ref{thm:n-players-medium-peak-uniform} and~\ref{thm:n-players-large-peak-uniform} not only minimize the number of documents that need to be added to achieve an equilibrium but also maximize the users' social welfare.

\begin{restatable}{lemma}{SocialWelfareUniform}
\label{lem:uniform-social-welfare-maximization}
    The equilibria achieved in Theorems~\ref{thm:n-players-medium-peak-uniform} and~\ref{thm:n-players-large-peak-uniform} maximize both the users' social welfare and publishers' exposure. 
\end{restatable}

\section{Best-Response Dynamics}
\label{sec:best-response}

\begin{figure}
    \centering
    \begin{tikzpicture}[node distance=0.07\linewidth and -0.1\linewidth]
    \node[inner sep=0] (img1) 
    {\begin{subfigure}{0.25\linewidth}
        \centering
        \includegraphics[width=0.9\linewidth]{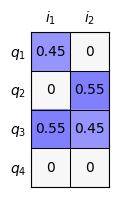}
        \caption{$s^0$}
        \label{subfig:fairness-example-s0}
    \end{subfigure}};
    \node[below left=of img1, inner sep=0] (img2)
    {\begin{subfigure}{0.25\linewidth}
        \centering
        \includegraphics[width=0.9\linewidth]{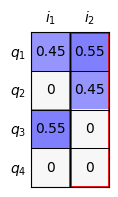}
        \caption{$s^1_{\text{unfair}}$}
        \label{subfig:fairness-example-s1-unfair}
    \end{subfigure}};
    \node[below right=of img1, inner sep=0] (img3)
    {\begin{subfigure}{0.25\linewidth}
        \centering
        \includegraphics[width=0.9\linewidth]{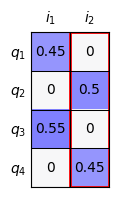}
        \caption{$s^1_{\text{fair}}$}
        \label{subfig:fairness-example-s1-fair}
    \end{subfigure}};

    \draw[->, thick] (img1) -- (img2);
    \draw[->, thick] (img1) -- (img3);
    \end{tikzpicture}
    
    \caption{Possible fair and unfair improvement steps in the game $G=\langle 2, 4, 1, 0.4 \rangle$, with the player that deviates marked in red.}
    \label{fig:fairness-importance-example}
\end{figure}

Best-response dynamics refers to the process in which players update their strategies iteratively when they have a profitable deviation from their current strategy. 
Specifically, at each time step, an arbitrary player among those who can improve their outcome selects a strategy that maximizes their utility. These dynamics continue until a state of equilibrium is reached, where no player has a profitable deviation. 
    
    In the absence of external intervention, such as corpus enrichment, best response dynamics are known to potentially fail in converging \cite{Haya1arxiv}, even when pure equilibrium exists. 
In this section, we prove that the best response dynamics do converge when using the corpus enrichment method for $\frac{1}{\floor{\frac{m}{n}}+1} < p$.

The \emph{best response} is a strategy that yields the highest payoff for a player, given the strategies chosen by other players. 
Formally, a strategy $d_i'$ is deemed a \emph{best response} of player $i$ if it satisfies $U_i(s_{-i}, d_i') \ge \max_{d_i \in D} U_i(s_{-i}, d_i)$.
A \emph{best response improvement step} is a profile $(s_{-i}, d_i')$, where $d_i'$ is the best response of player $i$.
Henceforth, we will refer simply to an improvement step when discussing best response improvement steps.
An \emph{improvement path} is defined as a sequence of such steps, denoted $\{s^l\}_{l=0}$.
In each step $l$, exactly one player deviates from his action in the previous step, $s^{l-1}$, and this player is indicated by $i_l$.  
The players making these deviations may vary across different steps.

We introduce a mild assumption to address tie-breaking when players have multiple viable deviation strategies to select from.

\begin{definition}[Deviation Equity]
    For any player $i$ and deviation $d_i'$, the {\em deviation equity} is defined as the utility of the worst-off player among all those whose utilities are decreased by the deviation. 
    Formally, 
    \begin{equation*}
        v(d_i^{'}, s) = \min_{i': U_{i'}(d_i^{'}, s_{-i}) < U_{i'}(s)} U_{i'}(d_i^{'}, s_{-i})
    \end{equation*}
    If no player was harmed by $i$'s deviation to $d_i^{'}$, $v(d_i^{'}, s)=\infty$.
\end{definition}

\begin{definition}[Fair Player]
    A player $i$ is \emph{\fairPlayer} if given multiple possible best response strategies, he selects the deviation $d_{i}^{br}$ that maximizes the \emph {deviation equity}.
\end{definition}

The fairness assumption is made to avoid trivial divergence as the one depicted in Figure~\ref{fig:fairness-importance-example}. 
In the example, two distinct improvement steps are considered, both resulting in the same utility gain for the deviating player.
When the player opts for an "unfair" strategy, the resulting strategy profile (Figure ~\ref{subfig:fairness-example-s1-unfair}) merely permutes the initial strategy profile (Figure ~\ref{subfig:fairness-example-s0}), and therefore convergence is not guaranteed. 
In contrast, when a fair strategy is employed, the resulting strategy profile (Figure~\ref{subfig:fairness-example-s1-fair}) is a pure equilibrium.

Under the fairness assumption, if $m \pmod{n}=0$, 
any best response deviation from a strategy by a player does not decrease the utility of others, leading to convergence in dynamics.

\begin{restatable}{theorem}{BestResponseSimple}
\label{thm:br-convergence-simple}
    Consider a game $G=\langle n, m, p, t \rangle$ where $m \pmod{n}=0$ and $p>\frac{1}{\floor{\frac{m}{n}}+1}$.
    Any sequence of best response dynamics $\{s^l\}_{l=0}$ where all players are fair players converges to a pure equilibrium $s^*$ in no more than $n$ steps.
\end{restatable}

To prove convergence in the general case, we proceed by identifying two potential patterns for the deviation of $i_l$:
\begin{itemize}
    \item \textbf{\devTypeZ:} $i_l$ deviates to a strategy in which he is the sole winner in $\floor{\frac{m}{n}}$ queries. 
    \item \textbf{\devTypeZplusOne:} $i_l$ deviates to a strategy where he wins in $\floor{\frac{m}{n}}+1$ queries, with at least $\floor{\frac{m}{n}}-1$ of them uniquely.
\end{itemize}

In the following lemma we demonstrate that for any given strategy profile $s$, each player's best response strategy falls into one of these two patterns.
Let $x_{i}(s) = \{q_j \in Q: w_j(s_{-i})<t\}$ be the set of queries in which, under $s_{-{i}}$, there is no document published by an original player whose score surpasses the threshold.

\begin{restatable}{lemma}{BestResponsePattern}
\label{lem:br-pattern}
    Let $G=\langle n, m, p, t\rangle$ with $t \ge \frac{1}{\floor{\frac{m}{n}}+1}$.
    Given any strategy profile $s$, the best response of each player $i$ is either of \devTypeZ or \devTypeZplusOne.
    Specifically, if $|x_i(s)|  < \floor{\frac{m}{n}} - 1$ then the best response is of \devTypeZ; 
    if $|x_i(s)| > \floor{\frac{m}{n}} - 1$ it is of \devTypeZplusOne; and if $|x_i(s)|  = \floor{\frac{m}{n}} - 1$ it can be of either type. 
\end{restatable}

Next, we demonstrate that if, during a sequence of fair best responses, a player deviates for the second time, then no other player that did not have a beneficial deviation at this point, will have a beneficial deviation in the next step.
\begin{restatable}{lemma}{DeviationEquityBoundSecondDev}
\label{lem:deviation-equity-bound-second-dev}
    Let $G=\langle n, m, p, t\rangle$ with $t \ge \frac{1}{\ceil{\frac{m}{n}}}$.
    If player $i_l$ has previously deviated, then subsequent deviations by $i_l$ ensure that no player who did not have a profitable deviation before will obtain one after this deviation.
\end{restatable}

The proof idea relies on the fact that after a first deviation, any subsequent deviations are of \devTypeZplusOne, which satisfies this property. 
Hence, a player's deviation can trigger future deviation by another player at most once, which is key to the main Theorem's proof.

\begin{restatable}{theorem}{BestResponseConvergence}
\label{thm:br-convergence}
    Let $G=\langle n, m, p, t \rangle$ where $m\pmod{n} >0$, $p > \frac{1}{\floor{\frac{m}{n}}+1}$ and $t \ge \frac{1}{\ceil{\frac{m}{n}}}$.
    If all players are \fairPlayer, any sequence of best-response dynamics $\{s^l\}_{l=0}$ will converge to a pure equilibrium $s^*$ with rate $O(\min\{m, n^2\})$.
\end{restatable}

For lower peak values, 
best-response dynamics do not necessarily converge. 
Consider, for instance, $G=\langle 2, 3, \frac{1}{2}, \frac{1}{2}-\epsilon \rangle$, where $\epsilon>0$.
Starting from an initial strategy profile  $s_0 = ((\frac{1}{2} - \epsilon,0, \frac{1}{2}), (0,\frac{1}{2}, \frac{1}{2} - \epsilon))$, a possible deviation might lead to the profile $s_1 = ((\frac{1}{2} - \epsilon,0, \frac{1}{2}), (\frac{1}{2}, \frac{1}{2} - \epsilon, 0))$. 
As \( s_1 \) is a permutation of \( s_0 \), the dynamics result in a cycle.

\subsection{Computationally Efficient Best Response Dynamics}

The strategy space is continious, so despite narrowing down the strategy patterns for any best-response deviation, there can still be infinitely many possible strategies.
We therefore provide an efficient algorithm that calculates a best-response strategy,
and show that using this algorithm, the fairness assumption can be relaxed, and that this relaxation is enough to obtain convergence.

\begin{algorithm}
\caption{(Almost) fair best response strategy construction for player $i_l$}
\begin{algorithmic}
    \STATE Initialize $d_{br}^j=0$ for all $q_j\in Q$.
    \STATE Set $d_{br}^j = t$ for all $q_j \in x_{i_l}(s)$.
    \IF {$|x_{i_l}(s)| \ge \floor{\frac{m}{n}}-1$}
        \STATE numQueriesToAdd = $\floor{\frac{m}{n}}+1 - |x_{i_l}(s)|$
        \STATE $\epsilon = 0$
    \ELSE
        \STATE numQueriesToAdd = $\floor{\frac{m}{n}} - |x_{i_l}(s)|$
        \STATE $\epsilon = \frac{1}{\floor{\frac{m}{n}}(\floor{\frac{m}{n}}+1)}$
    \ENDIF

    \FOR {$1$ \textbf{to} numQueriesToAdd}
        \STATE $i_{\max} = \arg\max_{i \neq i_l, |J_i^{tie}(s_{-i_l}, d_{br}) \le 1} |J_i(s_{-i_l}, d_{br})|$
        \STATE $q_{j_{\text{add}}}$ = arbitrary element of $J_{i_{\max}}^{\solo}(s_{-i_l}, d_{br})$
        \STATE $d_{br}^{j_{\text{add}}} = t + \epsilon$
    \ENDFOR

    \RETURN $d_{br}$
\end{algorithmic}
\label{alg:build-almost-fair-br}
\end{algorithm}

Suppose a player can access the identity of original documents with highest score for each query, as well as their score. 
It is therefore easy to see that algorithm~\ref{alg:build-almost-fair-br} runs in linear time and outputs a best-response strategy $d_{i_l}^{alg}$.
It remains to show that this best-response startegy approximates deviation equity in a way that suffices to obtain our main result.
This is guaranteed by the following Lemma:  
\begin{restatable}{lemma}{DeviationEquityBound}
\label{lem:deviation-equity-bound}
    Let $G=\langle n, m, p, t\rangle$ with $t \ge \frac{1}{\floor{\frac{m}{n}}}$.
    If $i_l$ performs a deviation of \devTypeZ, then 
    \begin{equation*}
        \floor{\frac{m}{n}} - 1 < v(d_{i_l}^{alg}, s^l) \le \max_{d_{i_l}^{br}} v(d_{i_l}^{br}, s^l) 
    \end{equation*}
    and if the deviation is of \devTypeZplusOne, then
    \begin{equation*}
        \floor{\frac{m}{n}} - \frac{1}{2} < v(d_{i_l}^{alg}, s^l) \le \max_{d_{i_l}^{br}} v(d_{i_l}^{br}, s^l) 
    \end{equation*}
\end{restatable}

As it can be shown Lemma~\ref{lem:deviation-equity-bound} implies Lemma~\ref{lem:deviation-equity-bound-second-dev} we get the desired result.

\section{Generalized Corpus Extension}
\label{sec:general}

\begin{figure}

    \centering
    \begin{subfigure}{0.4\linewidth}
        \centering
        \includegraphics[width=0.9\linewidth]{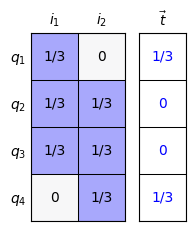}
        \caption{$\|\Vec{t}\|_1=\frac{2}{3}$}
        \label{subfig:sw-general-example-game1}
    \end{subfigure}%
    \begin{subfigure}{0.4\linewidth}
        \centering
        \includegraphics[width=0.9\linewidth]{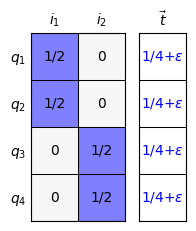}
        \caption{$\|\Vec{t}\|_1= 1 +  4 \cdot \epsilon$}
        \label{subfig:sw-general-example-game2}
    \end{subfigure}
    \caption{Illustration of different equilibrium profiles in the game $G_{\generalThreshold}=\langle 2, 4, 1, \generalThreshold\rangle$ for varying $\|\generalThreshold\|_1$ values.}
    \label{fig:social-welfare-general-example}
\end{figure}

\begin{figure*}
    \centering
    \begin{subfigure}{0.18\linewidth}
        \centering
        \tikz[remember picture] \node[inner sep=0pt] (figs0) {\includegraphics[width=\linewidth]{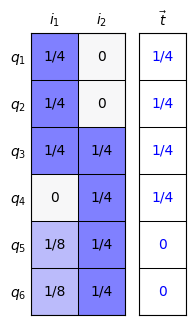}};
        \caption{$s_0$}
        \label{fig:s0}
    \end{subfigure}%
    \hspace{8mm}
    \begin{subfigure}{0.18\linewidth}
        \centering
        \tikz[remember picture] \node[inner sep=0pt] (figs1) {\includegraphics[width=\linewidth]{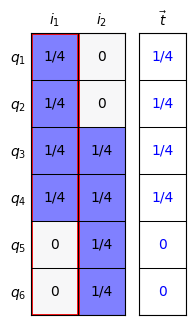}};
        \caption{$s_1$}
        \label{fig:s1}
    \end{subfigure}%
    \hspace{8mm}
    \begin{subfigure}{0.18\linewidth}
        \centering
        \tikz[remember picture] \node[inner sep=0pt] (figs2) {\includegraphics[width=\linewidth]{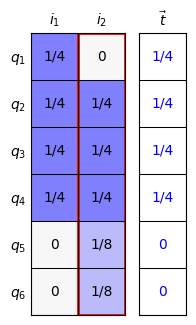}};
        \caption{$s_2$}
        \label{fig:s2}
    \end{subfigure}
    \hspace{8mm}
    \begin{subfigure}{0.18\linewidth}
        \centering
        \tikz[remember picture] \node[inner sep=0pt] (figs3) {\includegraphics[width=\linewidth]{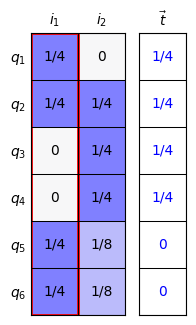}};
        \caption{$s_3$}
        \label{fig:s3}
    \end{subfigure}
    \caption{Improvement cycle in the game $G=\langle 2, 6, 1, (\frac{1}{4}, \frac{1}{4}, \frac{1}{4}, \frac{1}{4}, 0, 0)\rangle$, with deviations marked in red at each step.}
    \label{fig:general-best-response-example}
    
    \begin{tikzpicture}[overlay, remember picture]
        \draw[->, thick] (figs0.east) -- (figs1.west);
        \draw[->, thick] (figs1.east) -- (figs2.west);
        \draw[->, thick] (figs2.east) -- (figs3.west);
    \end{tikzpicture}
\end{figure*}
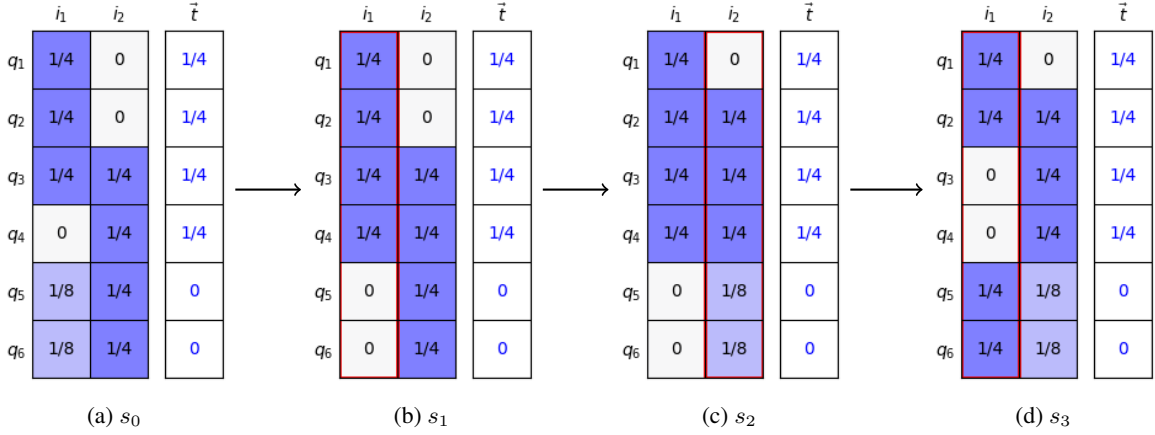

In the following section, we extend the results from Section~\ref{sec:uniform} by allowing different thresholds for distinct queries. 
This relaxation allows for achieving equilibrium with fewer added documents. 
For example, let $G_{\generalThreshold}=\langle 4, 5, 1, \generalThreshold\rangle$.
According to Theorem~\ref{thm:n-players-large-peak-uniform}, an equilibrium profile in the game with uniform corpus extension will be
\begin{equation*}
[s | t] = 
\left[
\begin{array}{cccc|c}
0.5 & 0.5 & 0 & 0 & 0.5\\
0.5 & 0.5 & 0 & 0 & 0.5 \\
0 & 0 & 0.5 & 0 & 0.5 \\
0 & 0 & 0.5 & 0.5 & 0.5 \\
0 & 0 & 0 & 0.5 & 0.5 \\
\end{array}
\right]
\end{equation*}

In the game above, using a uniform threshold requires the addition of three documents to reach equilibrium. 
However, considering that ties in query scores establish a "threshold" by themselves, we can optimize this by setting some scores to zero.
The revised equilibrium profile, which minimizes the number of added documents, is as follows:

\begin{equation*}
[s | \generalThreshold] = 
\left[
\begin{array}{cccc|c}
0.5 & 0.5 & 0 & 0 & 0\\
0.5 & 0.5 & 0 & 0 & 0 \\
0 & 0 & 0.5 & 0 & 0.5 \\
0 & 0 & 0.5 & 0.5 & 0 \\
0 & 0 & 0 & 0.5 & 0.5 \\
\end{array}
\right]
\end{equation*}

Employing non-uniform thresholds necessitates the addition of just one document, demonstrating a significant reduction in the number of dummy documents needed to achieve pure equilibrium existence. 
We proceed to analyze games with non-uniform thresholds.

\subsection{Equilibrium}
\label{sec:general-equilibrium}

The following theorems are generalizations of the corresponding theorems from the uniform case, which we discussed in Section~\ref{sec:uniform}.

\begin{restatable}{theorem}{MediumPeakGeneral}
\label{thm:n-players-medium-peak-general}
    $G_{\generalThreshold}=\langle n, m, p, \generalThreshold\rangle$ with $\frac{1}{\ceil{\frac{2m}{n}} - 1} < p \le \frac{1}{\floor{\frac{m}{n}}+1}$
    has a pure Nash equilibrium iff
    \begin{equation*}
        \| \generalThreshold \|_1 > n \cdot (1-p) - 2p \cdot (n\pFloor - m)
    \end{equation*}
\end{restatable}

\begin{restatable}{theorem}{LargePeakGeneral}
\label{thm:n-players-large-peak-general}
$G_{\generalThreshold}=\langle n, m, p, \generalThreshold\rangle$ where $\frac{1}{\floor{\frac{m}{n}}+1} < p$ has a pure Nash equilibrium iff there exists some $\epsilon>0$ such that 
    \begin{equation*}
        \|\Vec{t}\|_1 \ge \begin{cases}
            \frac{m-n}{\frac{m}{n}+1}, & m \mymod{n} = 0 \\
            \frac{m-n}{\ceil{\frac{m}{n}}} + \epsilon, & 1 \le m \mymod{n} \le n-2 \\
            n \cdot \left(1 - \frac{1}{\ceil{\frac{m}{n}}} \right) + \epsilon, & m \mymod{n} = n-1 \\
        \end{cases}
    \end{equation*}
\end{restatable}

Since the generalized corpus extension method is a generalization of the uniform corpus extension method, we obtain:
\begin{corollary}
    Let $n\in \mathbb{N}$. 
    For any $m \in \mathbb{N}$ and $0 <p\le 1$ there exists some $\generalThreshold\in \mathbb{R}^m$ such that $\|\generalThreshold\|_1 \le n$ and
    the game $G_{\generalThreshold}=\langle n, m, p, \generalThreshold\rangle$ has a pure equilibrium.
\end{corollary}

Comparing the threshold norms' bounds between uniform and generalized corpus extension methods shows that the difference in the lower bounds is most pronounced in games with a low query-to-player ratio. 
As this ratio increases, the difference diminishes, leading to the lower bounds on thresholds of the generalized method converging to those of the uniform method.

\subsection{Social Welfare}

In most instances, the equilibrium profile observed in the game with the general corpus extension method matches that obtained in the uniform setting. Nevertheless, there are specific situations where the equilibria differ, which impacts social welfare.
Figure~\ref{fig:social-welfare-general-example} illustrates an example of these differences by comparing two equilibrium profiles that arise under two different thresholds allocations, showing that social welfare may not be maximized in equilibrium when thresholds are non-uniform. 
This comparison highlights the tradeoff between minimizing the number of documents added to the corpus and maximizing social welfare.

\subsection{Best-Response Dynamics}

In Section~\ref{sec:best-response} we showed that the game with uniform corpus extension has a desired property: best-response dynamics convergence. 
However, this property does not extend to the generalized method. Figure~\ref{fig:general-best-response-example} presents an example from a game with general corpus enrichment featuring a non-uniform thresholds vector. As $s_3$ and $s_0$ are permutations of each other, the dynamics fail to converge.

\section{Discussion}

Our work introduced corpus enrichment as a technique for obtaining stable search and recommendation ecosystems. The resulting output maximizes social welfare, while only original content providers are selected, and the number of dummy documents needed is small. Moreover, we showed that for uniform corpus enrichment, this can be obtained by best-response dynamics. Generalized corpus enrichment allows to further restrict the number of dummy documents to be added to obtain stable output. 

A limitation of our results is that there are restrictions on the settings in which best-response dynamics converge. In particular, best-response dynamics does not always converge under the generalized corpus enrichment technique. While as discussed in Section \ref{sec:best-response} our results can be obtained under some minimal observability assumptions of winning documents and their score, it is still a significant limitation, which may be of interest to further relax. Further work may consider some restrictions on the best-response dynamics beyond the notion of fairness that may suffice for convergence. One other possible extension of our study is to consider corpus enrichment strategies in which social welfare does increase due to the introduction of few high-quality documents while relaxing the requirement that only content written by the original authors is selected.

\section*{Acknowledgements} 
This work was supported by funding from the European Research Council (ERC) under the European Union's Horizon 2020 research and innovation programme (grant agreement 740435).

\bibliography{refs/refs, refs/cunlp-ir,refs/fiana-bib}

\begin{thebibliography}{17}
\providecommand{\natexlab}[1]{#1}

\bibitem[{Ben{-}Basat, Tennenholtz, and Kurland(2017)}]{basat2017game}
Ben{-}Basat, R.; Tennenholtz, M.; and Kurland, O. 2017.
\newblock A game theoretic analysis of the adversarial retrieval setting.
\newblock \emph{Journal of Artificial Intelligence Research}, 60: 1127--1164.

\bibitem[{Ben-Porat et~al.(2018)Ben-Porat, Goren, Rosenberg, and Tennenholtz}]{FacilityLocationRec}
Ben-Porat, O.; Goren, G.; Rosenberg, I.; and Tennenholtz, M. 2018.
\newblock From Recommendation Systems to Facility Location Games.
\newblock In \emph{AAAI Conference on Artificial Intelligence}. AAAI Press.
\newblock ISBN 978-1-57735-809-1.

\bibitem[{Ben{-}Porat, Rosenberg, and Tennenholtz(2019)}]{LearnignDynamicsPaper}
Ben{-}Porat, O.; Rosenberg, I.; and Tennenholtz, M. 2019.
\newblock Convergence of Learning Dynamics in Information Retrieval Games.
\newblock \emph{Proceedings of the AAAI Conference on Artificial Intelligence}, 33(01): 1780--1787.

\bibitem[{Ben-Porat and Tennenholtz(2018)}]{ben2018game}
Ben-Porat, O.; and Tennenholtz, M. 2018.
\newblock A game-theoretic approach to recommendation systems with strategic content providers.
\newblock \emph{Advances in Neural Information Processing Systems}, 31.

\bibitem[{Borel(1921)}]{borel1921}
Borel, E. 1921.
\newblock La théorie du jeu et les équations intégrales à noyau symétrique. Comptes rendus de l’Académie des Sciences, 173 (1304-1308):58.

\bibitem[{Chen et~al.(2023{\natexlab{a}})Chen, He, Sun, and Sun}]{Chen+al:23a}
Chen, X.; He, B.; Sun, L.; and Sun, Y. 2023{\natexlab{a}}.
\newblock Defense of Adversarial Ranking Attack in Text Retrieval: Benchmark and Baseline via Detection.
\newblock \emph{CoRR}, abs/2307.16816.

\bibitem[{Chen et~al.(2023{\natexlab{b}})Chen, He, Ye, Sun, and Sun}]{Chen+al:23b}
Chen, X.; He, B.; Ye, Z.; Sun, L.; and Sun, Y. 2023{\natexlab{b}}.
\newblock Towards Imperceptible Document Manipulations against Neural Ranking Models.
\newblock In \emph{Findings of the Association for Computational Linguistics: {ACL}}, 6648--6664.

\bibitem[{Eilat and Rosenfeld(2023)}]{EilatR23}
Eilat, I.; and Rosenfeld, N. 2023.
\newblock Performative Recommendation: Diversifying Content via Strategic Incentives.
\newblock In Krause, A.; Brunskill, E.; Cho, K.; Engelhardt, B.; Sabato, S.; and Scarlett, J., eds., \emph{International Conference on Machine Learning, {ICML} 2023, 23-29 July 2023, Honolulu, Hawaii, {USA}}, volume 202 of \emph{Proceedings of Machine Learning Research}, 9082--9103. {PMLR}.

\bibitem[{Goren et~al.(2021)Goren, Kurland, Tennenholtz, and Raiber}]{Goren_2021}
Goren, G.; Kurland, O.; Tennenholtz, M.; and Raiber, F. 2021.
\newblock Driving the Herd: Search Engines as Content Influencers.
\newblock In \emph{Proceedings of the 30th ACM International Conference on Information and Knowledge Management}. ACM.

\bibitem[{Hron et~al.(2023)Hron, Krauth, Jordan, Kilbertus, and Dean}]{HronKJKD23}
Hron, J.; Krauth, K.; Jordan, M.~I.; Kilbertus, N.; and Dean, S. 2023.
\newblock Modeling content creator incentives on algorithm-curated platforms.
\newblock In \emph{The Eleventh International Conference on Learning Representations, {ICLR} 2023, Kigali, Rwanda, May 1-5, 2023}. OpenReview.net.

\bibitem[{Jagadeesan, Garg, and Steinhardt(2023)}]{Jagadeesan0S23}
Jagadeesan, M.; Garg, N.; and Steinhardt, J. 2023.
\newblock Supply-Side Equilibria in Recommender Systems.
\newblock In Oh, A.; Naumann, T.; Globerson, A.; Saenko, K.; Hardt, M.; and Levine, S., eds., \emph{Advances in Neural Information Processing Systems 36: Annual Conference on Neural Information Processing Systems 2023, NeurIPS 2023, New Orleans, LA, USA, December 10 - 16, 2023}.

\bibitem[{Kurland and Tennenholtz(2022)}]{compSearch}
Kurland, O.; and Tennenholtz, M. 2022.
\newblock Competitive search.
\newblock In \emph{Proceedings of the 45th International ACM SIGIR Conference on Research and Development in Information Retrieval}, 2838--2849.

\bibitem[{Liu et~al.(2023)Liu, Zhang, Guo, de~Rijke, Chen, Fan, and Cheng}]{Liu+al:23a}
Liu, Y.; Zhang, R.; Guo, J.; de~Rijke, M.; Chen, W.; Fan, Y.; and Cheng, X. 2023.
\newblock Topic-oriented Adversarial Attacks against Black-box Neural Ranking Models.
\newblock In \emph{Proceedings of SIGIR}, 1700--1709.

\bibitem[{Mladenov et~al.(2020)Mladenov, Creager, Ben-Porat, Swersky, Zemel, and Boutilier}]{mladenov2020optimizing}
Mladenov, M.; Creager, E.; Ben-Porat, O.; Swersky, K.; Zemel, R.; and Boutilier, C. 2020.
\newblock Optimizing long-term social welfare in recommender systems: A constrained matching approach.
\newblock In \emph{International Conference on Machine Learning}, 6987--6998. PMLR.

\bibitem[{Nachimovsky et~al.(2024)Nachimovsky, Tennenholtz, Raiber, and Kurland}]{Haya1arxiv}
Nachimovsky, H.; Tennenholtz, M.; Raiber, F.; and Kurland, O. 2024.
\newblock Ranking-Incentivized Document Manipulations for Multiple Queries.
\newblock In \emph{Proceedings of the 2024 ACM SIGIR International Conference on Theory of Information Retrieval}, ICTIR '24, 61–70. New York, NY, USA: Association for Computing Machinery.
\newblock ISBN 9798400706813.

\bibitem[{Raifer et~al.(2017)Raifer, Raiber, Tennenholtz, and Kurland}]{raifer2017information}
Raifer, N.; Raiber, F.; Tennenholtz, M.; and Kurland, O. 2017.
\newblock Information retrieval meets game theory: The ranking competition between documents' authors.
\newblock In \emph{Proceedings of the 40th International ACM SIGIR Conference on Research and Development in Information Retrieval}, 465--474.

\bibitem[{Yao et~al.(2023)Yao, Li, Sankararaman, Liao, Zhu, Wang, Wang, and Xu}]{YaoLSLZWWX23}
Yao, F.; Li, C.; Sankararaman, K.~A.; Liao, Y.; Zhu, Y.; Wang, Q.; Wang, H.; and Xu, H. 2023.
\newblock Rethinking Incentives in Recommender Systems: Are Monotone Rewards Always Beneficial?
\newblock In Oh, A.; Naumann, T.; Globerson, A.; Saenko, K.; Hardt, M.; and Levine, S., eds., \emph{Advances in Neural Information Processing Systems 36: Annual Conference on Neural Information Processing Systems 2023, NeurIPS 2023, New Orleans, LA, USA, December 10 - 16, 2023}.

\end{thebibliography}

\newpage
\appendix

\section{Omitted Proofs from Section~\ref{sec:model}}

\Translation*
\begin{proof}
    Let $\eta \in \mathbb{N}$.
    Assume that there exists a set of documents $T$ such that $|T| = \eta$ and the game $G_T = \langle n, m, p, T \rangle$ has a pure equilibrium.
    Denote by $\Vec{t}_j$ the value of the document in $T$ corresponding to the highest ranking score with respect to query $q_j$, i.e. $\Vec{t}_j \definedas \arg_{d_t \in T} \max {r(d_t, q_j)}$ for every $q_j \in Q$.
    Then $\|\Vec{t} \|_1 = \sum_{q_j \in Q} \Vec{t}_j \le \sum_{q_j \in Q} \sum_{d_t \in T} d_t^j \le \eta$, and the game $G_{\Vec{t}} = \langle n, m, p, \Vec{t} \rangle$ also has a pure equilibrium.
\end{proof}

\section{Omitted Proofs from Section~\ref{sec:uniform}}

\MediumPeakUniform*

\begin{proof}
    Let $k\in \mathbb{N}$ s.t. $\frac{1}{k+1} < p \le \frac{1}{k}$.
    
    In the original model \cite{Haya1arxiv}, there were two necessary conditions for the existence of an equilibrium:
    \begin{enumerate}
        \item The winning value must be $w_j(s)=p$ for every $q_j \in Q$.
        \item No player is the unique winner in more than one query. 
    \end{enumerate}

    For $p > \frac{1}{\ceil{\frac{2m}{n}}-1}$, there exists no strategy profile $s$ that satisfies both conditions, and therefore a pure equilibrium without corpus enrichment doesn't exist ~\cite{Haya1arxiv}. 
    
    For peak values in the range $\frac{1}{\ceil{\frac{2m}{n}}-1} < p \le \frac{1}{\ceil{\frac{m}{n}}}$, there exists a strategy profile $s$ for which the first condition holds, 
    but there will always be at least one player who is the unique winner in more than one query.

    It's easy to see that regardless of the threshold value, if there exists a query where $w_j(s)<p$, then every player can guarantee a utility of at least $k$. 
    Hence, the sum of utilities of all players would be at least $n\cdot k$.
    Therefore,the case where $n \cdot k > m$, $w_j(s)=p$ for every $q_j \in Q$ and $t\le p$.
    However, if $n \cdot k = m$, there exists a strategy profile where each player wins in $k+1$ queries uniquely, where $k-1$ are won solely by him and in the other $2$, there are 2 winners in each query. 
    For this strategy profile to be an equilibrium, it must hold that $t=\frac{1}{k+1}$.

    We will show that in the uniform case, there always exists a pure equilibrium profile where $w_1(s)=\ldots=w_m(s)=p$, and this strategy profile also minimizes the threshold.
    
    First, since $p \le \frac{1}{\ceil{\frac{m}{n}}}$ there exists a strategy profile $s$ where $w_1(s)=\ldots=w_m(s)=p$.
    Under $s$, every player is a winner in exactly $k$ queries, i.e., $|J_i(s)|=k$.
    If $s$ is an equilibrium, $h_j(s)\le 2$ for every query $q_j \in Q$, so there are exactly one or two winners in each query.
    Let $\beta_i(s)$ denote the number of queries in which player $i$ is the unique winner: $\beta_i(s) = |\{q_j \in J_i(s) : h_j(s)=1\}|$.
    We get that
    \begin{equation*}
    \begin{split}
        n \cdot k &= \sum_{q_j \in Q} h_j(s) \\
        &= \sum_{i=1}^n \beta_i(s) + 2 \cdot (m - \sum_{i=1}^n \beta_i(s)) \\
        &= 2\cdot m - \sum_{i=1}^n \beta_i(s)
    \end{split}
    \end{equation*}
    Therefore, the total number of queries where there is exactly one winner is 
    \begin{equation*}
        \sum_{i=1}^n \beta_i(s) = 2\cdot m - n \cdot k 
    \end{equation*}
    Denote $\beta_{max}(s) = \max_{1 \le i\le n}\beta_i(s)$.
    Then 
    \begin{equation}
    \label{eq:beta-max-lower-bound}
    \begin{split}
        \beta_{max}(s) &\ge \frac{1}{n} \sum_{i=1}^n \beta_i(s) = \frac{2m}{n}- k \\
        \Rightarrow \quad &  \beta_{max}(s) \ge \ceil{\frac{2m}{n}} - k
    \end{split}
    \end{equation}
    
    Since the winning value is $p$, any profitable deviation of a player would be one in which he becomes the winner in $k+1$ queries. 
    This will be achieved by publishing a document $d'$ defined as follows:

    \begin{equation*}
        {d_j'} = \begin{cases}
            p, & q_j \in \{q_j \in J_i(s) : h_j(s) \ge 2\} \\
            p, & q_j = q_{j_1} \\
            \max\{w_j(s_{-i})+\epsilon, t\}, & q_j \in \{q_j \in J_i(s) : h_j(s) = 1\}\\
            0, & \text{otherwise}\\
        \end{cases}
    \end{equation*}
    Where $q_{j_1}$ is some arbitrary query where player $i$ is not the current winner and $\epsilon>0$.
    If the deviation is well defined, i.e., $\sum_{q_j \in Q} {d_j'} \le 1$, then $i$ has a profitable deviation.
    However, if no player has such a deviation, then the profile $s$ is a pure equilibrium.

    Since $w_j(s_{-i}) \le 1 - k \cdot p$, if $t<w_j(s_{-i})+\epsilon$ then $\sum_{q_j \in Q} {d_j'} < 1$.
    Therefore, $\max\{w_j(s_{-i})+\epsilon, t\} = t$.

    In order for player $i$ to not be able to deviate profitably, it must hold that 
    \begin{equation}
    \label{eq:bound-on-t-with-beta-by-player}
        1 < \sum_{q_j \in Q} {d_j'} = (k - \beta_i(s)) \cdot p + \beta_i(s) \cdot t + p
    \end{equation}
    
    Thus, $s$ is a pure equilibrium only if
    \begin{equation}
    \label{eq:bound-on-t-with-beta-max}
        t > p - \frac{p \cdot (k+1) - 1}{\beta_{max}(s)}
    \end{equation}

    It's left to show that there exists some strategy profile $s^*$ for which $\beta_{max}(s^*) = \ceil{\frac{2m}{n}} - k$, which is the minimal possible value for $\beta_{max}(\cdot)$ according to Equation~\ref{eq:beta-max-lower-bound}.
    Such a strategy profile can be obtained, for example,  by running the Algorithm proposed for the original model~\cite{Haya1arxiv}.
    Hence, we get that the tightest lower bound on $t$ is 
    \begin{equation*}
        t > p - \frac{p \cdot (k+1) - 1}{ \ceil{\frac{2m}{n}} - k}
    \end{equation*}
\end{proof}

\LargePeakUniform*
\begin{proof}
    Let $z_1 = \floor{\frac{m}{n}}$.

    First, we show that if $t \ge \frac{1}{z_1+1}$, the game possesses a pure equilibrium $s$.
    Consider the profile $s=(d_1,\ldots, d_n)$, where each player publishes a document with $\frac{1}{z_1+1}$ for $z_1+1$ times. 
    There exists a strategy profile where each player is the unique winner in at least $z_1-1$ queries.
    This profile can be constructed, for example, by running the Algorithm proposed for the original model with $p'=\frac{1}{z_1+1}$ \cite{Haya1arxiv}.
    In the resulting profile, every player will be the unique winner in either $z_1$ or $z_1-1$ queries.
    It's easy to see that for $t = \frac{1}{z_1+1}$, this profile is a pure equilibrium. 
    Moreover, for every equilibrium where there exists a player $i$ such that $|J_i(s)|=z_1+1$, and $U_i(s) < z_1+1$, the bound on $t$ is tight.

    If $m \pmod{n} > 0$, there always exists such a player, so for $t < \frac{1}{z_1+1}$ there exists no pure equilibrium. 
    However, if $m \pmod{n} = 0$, the bound on $t$ can be further relaxed. 
    Specifically, there exists a strategy profile where every player is the unique winner in exactly $z_1$ queries. 
    
    A player $i$ will have a profitable deviation if she can become the winner in at least one more query in which she wasn't the winner before.
    Since the winning value $w_j(s)=\frac{1}{z_1}$ for all $q_j \in Q$, a player will have a profitable deviation iff
    \begin{equation*}
        z_1 \cdot t + \frac{1}{z_1} \le 1
    \end{equation*}
    Hence, we get that if $t> (1-\frac{1}{z_1})\cdot \frac{1}{z_1}$, then $s$ is a pure equilibrium. 
\end{proof}

\NDocsEquilibriumUniform*
\begin{proof}
    We showed that given any value of $m$ and $p$, there exists an equilibrium in which the winner is one of the publishers.
    Therefore,
    \begin{equation*}
        \sum_{j=1}^m t_j \le \sum_{j=1}^m w_j(s) \le \sum_{i=1}^n \sum_{j=1}^m d_i^j \le n \cdot 1
    \end{equation*}
    So there exists a document set $T$ such that $|T|\le n$ for which the game $G = \langle n, m, p, t \rangle$ has a pure equilibrium. 
\end{proof}

\SocialWelfareUniform*
\begin{proof} 
    For $\frac{1}{\ceil{\frac{2 \cdot m}{n}} - 1} < p \le \frac{1}{\ceil{\frac{m}{n}}}$, we saw that $w_j(s) = p$ for every query $q_j$, therefore the result is immediate. 

    For $\frac{1}{\ceil{\frac{m}{n}}} < p \le 1$, we have that $w_1(s) = \ldots = w_m(s) = \frac{1}{\ceil{\frac{m}{n}}}$. 
    Therefore $SW(s)=\frac{1}{\ceil{\frac{m}{n}}} - p$.

    For any profile $s'$, if there exists a player that wins in $x$ queries then there exists a query $q_{j_1}$ for which $w_{j_1}(s') \le \frac{1}{x}$.
    Therefore $SW(s') \le - |p - \frac{1}{x}|$.
    So for any profile $s'$ in which there exists a player that wins in more than $\ceil{\frac{m}{n}}$ queries, we get $SW(s')<SW(s)$.
    On the other hand, if all players win in less than $\ceil{\frac{m}{n}}$ queries, we get that there is at least one query for which there are no winners, so the social welfare is $-p$. 

    Therefore $\max_{s'} SW(s') \le \frac{1}{\ceil{\frac{m}{n}}} - p = SW(s)$.

\end{proof}
\section{Omitted Proofs from Section~\ref{sec:best-response}}

\BestResponseSimple*
\begin{proof}
    Let $z_1 = \frac{m}{n}$.
    Since $t>\frac{1}{z_1+1}$, for any strategy profile $s$ the number of queries in which a player is a winner is no more than $z_1$.
    Moreover, there will always be at least $z_1$ queries where $w_j(s^l_{-i_l}) < t$.
    Therefore any deviation step will be of \devTypeZ.
    So player $i_l$ can deviate without causing any other player's utility to decrease.
    Therefore after a player deviated once, he will no longer have a profitable deviation. 
    So each improvement path will be of a maximal length of $n$.
\end{proof}

We use an example to show that the bound on the threshold in theorem~\ref{thm:br-convergence-simple} is tight, i.e. if $t \le \frac{1}{\floor{\frac{m}{n}} + 1}$, the dynamics do not converge. 
Consider  $G=\langle n=2, m=4, p=1, t=\frac{1}{3} \rangle$, and $s^0 = ((\frac{1}{3}, \frac{1}{3}, \frac{1}{3}, 0), (0, 0, \frac{1}{2}, \frac{1}{3}))$. 
The following is a possible best response sequence (with \fairPlayer players): 
$s^1 = ((\frac{1}{3}, \frac{1}{3}, 0, \frac{1}{3}), (0, 0, \frac{1}{2}, \frac{1}{3}))$, 
$s^2 = ((\frac{1}{3}, \frac{1}{3}, 0, \frac{1}{3}), (0, 0, \frac{1}{3}, \frac{1}{2}))$.
However, note that $s^1$ is simply a permutation of $s^0$, therefore the dynamics can diverge.

We now proceed to prove the general case, where $m \pmod{n}>0$.

\BestResponsePattern*
\begin{proof}

    Let $i_{dev}$ represent the player deviating from the current strategy.
    Denote $z_1 = \floor{\frac{m}{n}}$ and $z_2 = m - n \cdot \floor{\frac{m}{n}}$.
    
    Since $t \ge \frac{1}{z_1+1}$, each document $d_i$ published by player $i$ must satisfy $\sum_{q_j \in Q} \indicator{d_i^j > t} \le z_1$.
    Consequently, there are at most $z_1 \cdot (n-1)$ queries where $w_j(s_{-i_{dev}}) > t$, 
    implying that in $s_{-i_{dev}}$ at least $m - z_1 \cdot (n-1)\ge z_1+z_2$ queries exist where $w_j(s_{-i_{dev}}) \le t$.
    Thus, $i_{dev}$ can always deviate to a strategy where they uniquely win in $z_1$ queries, corresponding to a deviation of \devTypeZ.

    Furthermore, since $t \ge \frac{1}{z_1+1}$, no player can win in more than $z_1 + 1$ queries. 
    We proceed to analyze the conditions under which $i_{dev}$ would prefer a deviation of \devTypeZplusOne:
    
    \begin{itemize}
        \item If $|x_{i_{dev}}(s)| \ge z_1 - 1$, $i_{dev}$ can publish a document in which he is the sole winner in exactly $|x_{i_{dev}}(s)|$ queries.
        \begin{itemize}
            \item If $|x_{i_{dev}}(s)| > z_1$, the utility is (strictly) higher than $z_1$, favoring a response of \devTypeZplusOne.
            \item If $|x_{i_{dev}}(s)| = z_1 - 1$, the utility is exactly $z_1$, making both \devTypeZ and \devTypeZplusOne possible deviation types.
        \end{itemize}
        \item If $|x_{i_{dev}}(s)| < z_1 - 1$, a strategy resulting in $|x_{i_{dev}}(s)| + 1$ wins would involve at least 3 queries where $i_{dev}$ does not win uniquely, yielding a utility below $z_1$.
        Conversely, a strategy where they uniquely win in $z_1$ queries yields a utility exactly equal to $z_1$. 
        Thus, the optimal strategy in this case is of \devTypeZ.
    \end{itemize}
\end{proof}

\DeviationEquityBoundSecondDev*

\begin{proof}
    Let $z_1 = \floor{\frac{m}{n}}$.
    Consider some player $i_l$ who performs a deviation of \devTypeZplusOne. 
    We argue that such a deviation will not create new profitable deviations for other players who previously had none. 
    This is because there always exists a player $i_{max}$ that achieves a utility greater than $z_1$. 
    If after any step $l$, $U_{i_{max}} < z_1$, then $i_{max}$ already had a profitable deviation prior to this step.

    Algorithm~\ref{alg:build-almost-fair-br} efficiently identifies strategies that maximize a player's utility. While these strategies do not necessarily maximize deviation equity, the deviation equity resulting from the algorithm serves as a lower bound for the maximal possible deviation equity. Therefore, we use the results from Lemma~\ref{lem:deviation-equity-bound} to establish that any subsequent deviations by player $i_l$ will be of \devTypeZplusOne if $i_l$ has deviated before.
    
    Let $l_{prev}$ be the last step in which $i_l$ deviated.
    According to Lemma~\ref{lem:br-pattern}, $i_l$’s prior deviation could have been either \devTypeZ or \devTypeZplusOne.
    
    If $i_l$'s previous deviation was of \devTypeZ, then immediately following that deviation at step $l_{prev}+1$, $i_l$’s utility was exactly $z_1$. 
    According to Lemma~\ref{lem:deviation-equity-bound}, $i_l$'s utility will not decrease in subsequent steps, maintaining $U_{i_l}(s^{l})=U_{i_l}(s^{l_{prev}})=z_1$.
    Therefore, if $i_l$ has a profitable deviation it means he can achieve a utility that exceeds $z_1$.
    Thus, every possible deviation for $i_l$ would be of \devTypeZplusOne.

    If $i_l$’s previous deviation was of \devTypeZplusOne, Lemma~\ref{lem:deviation-equity-bound} assures us that either $U_{i_l}(s) \ge z_1$ or $|x_{i_l}(s^l)| \ge z_1-1$. 
    If  $U_{i_l}(s) \ge z_1$, then as in the prior case, the only profitable deviation is \devTypeZplusOne.
    However, if $|x_{i_l}(s^l)| = z_1-1$ and $U_{i_l}(s)=z_1-\frac{1}{2}$, then both deviation types are theoretically  possible.
    Under the fairness condition, $i_l$ selects the one that maximizes the deviation equity, which is \devTypeZplusOne.
\end{proof}

\BestResponseConvergence*
\begin{proof}
    Let $\alpha_l$ denote the number of players who have deviated at least once by step $l$. 
    According to Lemma~\ref{lem:deviation-equity-bound-second-dev}, a deviation of \devTypeZplusOne by any player cannot induce a profitable deviation for players who hadn't had one before. 
    Consequently, the total number of deviations until a previously non-deviating player deviates for the first time is bounded by $\alpha_l + 1$.

    Denote by $l_i$ the step where $\alpha_{l} = i$ for the first time. 
    Setting $l_1=1$, we derive the following recursive formula: for every $1 \leq i \leq n$, $l_{i} \leq l_{i-1} + i$. 
    This implies that the overall process requires no more than $l_n + n$ steps for convergence, as no player will have a profitable deviation beyond this point.
    
    Recursively calculating $l_n$, we observe:
    \begin{equation}
    \label{eq:recursive-complexity}
        \begin{split}
            l_n &\le l_{n-1} + n \\
            &= (l_{n-2} + (n-1)) + n \\
            &\ldots \\
            &= l_1 + \sum_{i=2}^n i \\
            &= 1 + \sum_{i=2}^n i \\
            &= \frac{n(n+1)}{2} \\
        \end{split}
    \end{equation}
    Hence, the total steps until convergence, taking into account the final steps to reach stability, is $\frac{n(n+1)}{2} + n = \frac{n(n+3)}{2}$. 
    Therefore, the dynamics converge with complexity $O(n^2)$.

    If $\floor{\frac{m}{n}}< n$, we can potentially reduce the complexity of convergence.
    Denote $z_1 = \floor{\frac{m}{n}}$.
    In any deviation of \devTypeZ the number of affected players is at most $z_1$.
    Hence, the total number of deviations until a previously non-deviating player deviates for the first time is limited by $\min\{\alpha_l + 1, z_1\}$. 
    This constraint leads us to a modified recursive formula: for every $1 \leq i \leq n$, $l_{i} \leq l_{i-1} + \min\{i, z_1\}$.
    Following the same calculation as in Equation \ref{eq:recursive-complexity}, we get that $l_{z_1} \le \frac{z_1(z_1+1)}{2}$. 
    
    Beyond $z_1$, each subsequent $l_i$ only increases by $z_1$. 
    Therefore, we get:
    \begin{equation*}
        \begin{split}
            l_n &\le l_{n-1} + z_1 \\
            &\ldots \\
            &= l_{z_1} + (n-z_1) \cdot z_1 \\
            &\le \frac{z_1(z_1+1)}{2} + (n-z_1) \cdot z_1 \\
            &\le n \cdot z_1 \\
            &\le m
        \end{split}
    \end{equation*}
    Thus, the dynamics converge with a complexity of $O(m)$.
\end{proof}

\DeviationEquityBound*
\begin{proof}

We show that the deviation strategy obtained by Algorithm~\ref{alg:build-almost-fair-br} achieves the above lower bounds. 
In particular, we will show that the deviation can be performed greedily by adding one query every time to the set $J_{i_l}(s^{l+1})$.
Initially, $|J_{i_l}(s^{l+1})| = \min\{|x_{i_{l}}(s^l)|, z_1+1\}$.

In every iteration of the algorithm, the first step is to identify a player that wins in $z_1+1$ queries and is the unique winner in at least $z_1$ of them. 
It is therefore required to prove that such a player exists. 

It holds that $\max_{i \ne i_l}{U_i(s^l)} \ge z_1 + \frac{1}{n}$; otherwise, the sum of utilities will be less than $m$ - by contradiction.
Let $i_{max}$ denote a player that achieves this maximal utility. 
Note that $i_{max}$ is the unique winner in $z_1$ queries, and there is at most one query in which he is one of multiple winners. 

    If the deviation is of \devTypeZplusOne, $i_l$ will be one of two winners in $q_{add}$.
    Then $i_{max}$'s utility will now be:
    \begin{equation*}
         U_{i_{max}}(s^l_{-i_l}, d_{br}) = U_{i_{max}}(s^{l}) - \frac{1}{2} > z_1 - \frac{1}{2}
    \end{equation*}
    
    If the deviation is of \devTypeZ, $i_l$ will publish a document with a value strictly larger than $t$ in query $q_{add}$.
    Hence, the utility of $i_{max}$ will now decrease by $1$.
    So we get that 
    \begin{equation*}
         U_{i_{max}}(s^l_{-i_l}, d_{br}) = U_{i_{max}}(s^{l}) - \frac{1}{2} > z_1 - 1
    \end{equation*}

    This holds for every iteration; therefore, we get the result.
    Note that $i_{max}$ can repeat between different iterations, but the results still hold since the lower bound is independent of the exact utility. 
\end{proof}

\section{Omitted Proofs from Section~\ref{sec:general}}

\begin{figure*}[b]
    \centering
\includegraphics[width=\textwidth]{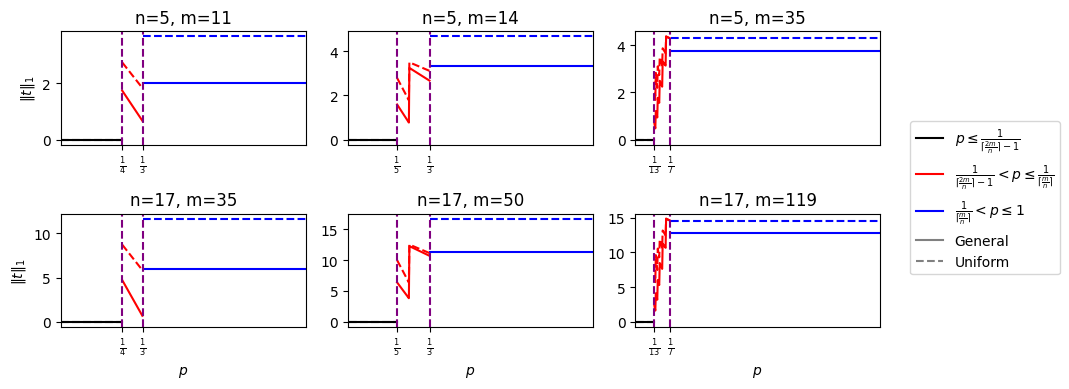}
    \caption{Estimated value of $\|t\|_1$ by different values of $p$.} 
    \label{fig:peak-threshold}
\end{figure*}

Figure~\ref{fig:peak-threshold} illustrates the differential impact of using uniform versus general thresholds on the number of documents required. 
The difference between the two methods is most noticeable in games where the ratio of queries to players is low. 
As this ratio increases, the difference between the two settings decreases, and the general thresholds converge towards the uniform threshold scenario.

\MediumPeakGeneral*
\begin{proof}

    We extend the analysis from the proof of Theorem~\ref{thm:n-players-medium-peak-uniform}, using the same notation. 
    In every strategy where $w_j(s)=p$ for all queries $q_j\in Q$, any profitable deviation of a player would be one in which he becomes the winner in $k+1$ queries.
    This can be achieved by publishing a document $d'$ defined as follows:
    
    \begin{equation*}
        {d_j'} = \begin{cases}
            p, & q_j \in \{q_j \in J_i(s) : h_j(s) \ge 2\} \bigcup \{q_{j_1}\} \\
            \Vec{t}_j , & q_j \in \{q_j \in J_i(s) : h_j(s) = 1\}\\
            0, & \text{otherwise}\\
        \end{cases}
    \end{equation*}

    Here, $q_{j_1}$ represents an arbitrary query where player $i$ is not currently winning, and $\epsilon > 0$. 
    Assuming that the deviation is feasible, i.e., $\sum_{q_j \in Q} {d_j'} \le 1$, player $i$ can profitably deviate. 
    If no player has such a deviation available, then strategy profile $s$ is a pure equilibrium.

    To ensure that player $i$ cannot profitably deviate, the following inequality must hold:
    \begin{equation*}
        1 < \sum_{q_j \in Q} {d_j'} = (k - \beta_i(s) + 1) \cdot p + \sum_{q_j \in J_i(s) : h_j(s) = 1} t_j
    \end{equation*}

    Therefore, $s$ is a pure equilibrium only if:
    \begin{equation*}
        \sum_{q_j \in J_i(s) : h_j(s) = 1} \Vec{t}_j > 1 - (k+1) \cdot p + \beta_i(s) \cdot p
    \end{equation*}

    In the case where thresholds are uniform across queries, the strategy that minimized $\|t\|_1$ was the strategy that minimized $\beta_{max}(s)$.
    For the general case, however, we demonstrate that the overall sum of thresholds remains constant irrespective of $\beta_i$ distribution, always equaling exactly $2m - nk$.
    
    Thus, the condition on the sum of thresholds can be derived as follows:
    \begin{equation*}
        \begin{split}
            \| \Vec{t} \|_1 &= \sum_{q_j\in Q} \Vec{t}_j \\
            &= \sum_{i=1}^n \sum_{q_j \in J_i(s) : h_j(s) = 1} \Vec{t}_j \\
            &> \sum_{i=1}^n  1 - (k+1) \cdot p + \beta_i(s) \cdot p \\
            &= n \cdot (1 - (k+1) \cdot p) + p \cdot\sum_{i=1}^n \beta_i(s) \\
            &= n \cdot(1 - (k+1) \cdot p) + p\cdot (2m - nk) \\
            &= n \cdot (1-p) - 2p \cdot (nk - m)
        \end{split}
    \end{equation*}

    The strategy profile demonstrated in the uniform scenario also achieves the minimum for this bound, confirming its tightness. 
\end{proof}

In scenarios where the threshold values are non-uniform, a key observation is that in any query $q_j$ with more than one winner, the threshold value $t_j$ can be $0$. 
Consequently, the minimal norm threshold vector, $\Vec{t}$, will have $t_j > 0$ if and only if there is exactly one winner in query $q_j$. 

Given that we consider only profiles where at least one of the highest-ranked documents originates from an actual player, we can define a specific measure for each player $i$ as follows:
\begin{equation*}
    \tau_i(s) \definedas \sum_{q_j \in J_i(s): h_j(s)=1} \Vec{t}_j
\end{equation*}

Consequently, the total norm of the threshold vector can be lower-bounded by the sum of these individual contributions:
\begin{equation*}
    \|\Vec{t}\|_1 \ge \sum_{i=1}^n \tau_i(s)
\end{equation*}

This relationship allows us to establish bounds on the threshold values by considering the possible strategic deviations of each player independently.

\LargePeakGeneral*

\begin{proof}
    Denote $z_1 = \floor{\frac{m}{n}}$ and $z_2 = m - n\cdot z_1$, so $m = n\cdot z_1 + z_2$.
    We begin the proof by deriving a lower bound on $\|\Vec{t}\|$.
    Assume $s$ is a pure equilibrium profile where, without loss of generality, $w_1(s) \ge \ldots \ge w_m(s)$.

    Define $n_{\phi}$ as the number of players for whom $\beta_i(s) = z_1+1-{\phi}$. 
    Analyzing the conditions under which a player $i$ has no profitable deviation, consider the following:

    For each player that is the sole winner in all of his queries, a profitable deviation would involve winning at least one additional query. 
    Therefore, to prevent profitable deviations, it must hold that:
    \begin{equation*}
        \begin{split}
            1 &< \tau_i(s) + \min_{q_j \notin J_i(s)} w_j(s) \\
        \end{split}
    \end{equation*}
    
    Rearranging, we obtain the following lower bound on $\tau_i(s)$:
        \begin{equation}
        \tau_i(s) >  1 - \min_{q_j \ne J_i(s)} w_j(s)
    \end{equation}

    For a player $i$ who does not win solely in all of his queries, let $q_{j_1}\in J_i(s)$ be a query where there are at least two winners. 
    Since an original player wins in every query, it follows that $w_{m}(s) < p$.
    If $w_{j_1}(s)=p$, then at least one of the players will have a profitable deviation to win uniquely in $w_m(s)$.
    Therefore $w_{j_1}(s) < p$.
    If there exists a query $q_j\in J_i(s)$ where $d_i^j > \max\{w_j(s_{-i}), \Vec{t}_j\}$, then player $i$ has a profitable deviation to a strategy in which he becomes the sole winner of query $q_{j_1}$, without changing his position for other queries. 
    Hence, we have:
    \begin{equation*}
        \begin{split}
            1 &=  \sum_{q_j \in Q} d_i^j \\
            & = \sum_{q_j \in J_i(s)} \max\{w_j(s_{-i}), \Vec{t}_j\}\\
            &= \sum_{q_j \in J_i(s) : h_j(s)=1} \Vec{t}_j + \sum_{q_j \in J_i(s) : h_j(s)\ge 2}  w_j(s_{-i}) \\
            &= \tau_i(s) + \sum_{q_j \in J_i(s) : h_j(s)\ge 2}  w_j(s_{-i}) \\
            &\ge \tau_i(s) + \sum_{q_j \in J_i(s) : h_j(s)\ge 2}  w_m(s) \\
        \end{split}
    \end{equation*}
    Rearranging, we get:
    \begin{equation}
        \tau_i(s) \ge 1 - \sum_{q_j \in J_i(s) : h_j(s)\ge 2} w_m(s)
    \end{equation}

    Hence, for $s$ to be an equilibrium:
    \begin{equation*}
        \begin{split}
            \| \Vec{t} \|_1 &\ge \sum_{i=1}^n \tau_i(s) \\
            &\ge n_0 \cdot \left(1 - w_m(s)\right) + n_0 \cdot \epsilon \\
            &\quad + n_1 \cdot \left(1 - w_m(s)\right) + n_2 \cdot \left(1 - 2 \cdot w_m(s)\right) \\
            &= n \cdot \left(1 - w_m(s)\right) - n_2 \cdot w_m(s) + n_0 \cdot \epsilon
        \end{split}
    \end{equation*}

    If $z_2>0$, it must hold that $\max_{1 \le i\le n} |J_i(s)| \ge z_1+1$, otherwise there will be a query where no original player wins.
    If $z_2=0$, there exists a strategy where $\max_{1 \le i\le n} |J_i(s)| = z_1$ - the one we showed in the uniform case.
    However, we will show that a strategy where $\max_{1 \le i\le n} |J_i(s)| \ge z_1+1$ exists in this case too, and it achieves a better bound on $t$.

    If $\max_{1 \le i \le n} |J_i(s)| \ge z_1+1$, then $w_m(s) \le \frac{1}{z_1+1}$. 
    Plugging in the above gives:
    \begin{equation}
    \label{eq:bound-t-norm}
        \| \Vec{t} \|_1 \ge n \cdot \left(1 - \frac{1}{z_1+1}\right) - n_2 \cdot \frac{1}{z_1+1} + n_0 \cdot \epsilon
    \end{equation}

    We demonstrate equilibrium profiles meeting this condition by implementing a strategy where $w_j(s) = \frac{1}{z_1+1}$ for all $q_j \in Q$. 
    The number of queries in which there are two winners is independent of the exact arrangement of $s$, and is $n-z_2$.

    If $0 \le z_2 \le n-2$, a strategy profile $s$ exists such that $z_2$ players each uniquely win in exactly $z_1+1$ queries, while $n-z_2$ players each uniquely win in $z_1-1$ queries, and two other queries resulting in a tie between two players. 
    Consequently, it follows that $n_2 = n-z_2$, $n_1=0$, and $n_0 = z_2$. 
    If $z_2 = n-1$, only one query results in a tie between two winners, leading to $n_1=2$ and $n_0 = n-2$.

    Upon applying these strategies and configurations, we confirm through Equation~\ref{eq:bound-t-norm} that these profiles indeed achieve the lower bound established for the norm of $\Vec{t}$.

\end{proof}

\end{document}